\documentclass[journal]{IEEEtran}
\usepackage{amsmath} 
\usepackage{amssymb}  

\usepackage{amsthm}  
\usepackage{amsmath,lipsum}
\usepackage{cuted}

\usepackage{mathtools}  
\usepackage{url}
\usepackage{cite}

\DeclareMathOperator{\diag}{diag}
\DeclareMathOperator*{\minimize}{minimize}
\DeclareMathOperator*{\subjectto}{subject\ to}
\DeclareMathOperator{\topin}{in}
\DeclareMathOperator{\topout}{out}

\usepackage{graphicx}
\usepackage[mathscr]{eucal}
\renewcommand{\i}[1]{\mathcal{#1}}

\usepackage{verbatim}
\usepackage{bm}
\usepackage{bbm}
\usepackage{color}
\usepackage{url}

\definecolor{forestgreen}{rgb}{0.33,0.61,0.34}
\definecolor{deepmagenta}{rgb}{0.8, 0.0, 0.8}
\definecolor{harvardcrimson}{rgb}{0.79, 0.0, 0.09}

\newcommand{\add}[1]{\textcolor{blue}{#1}}

\theoremstyle{definition}
\newtheorem{definition}{Definition}[section]
\newtheorem{asm}[definition]{Assumption}
\newtheorem{lemma}[definition]{Lemma}

\newtheorem{theorem}[definition]{Theorem}
\newtheorem{problem}[definition]{Problem}

\ifCLASSINFOpdf
\else
\fi
\allowdisplaybreaks

\begin{document}
%
\title{Optimization of Stochastic Switching Buffer Networks via DC Programming}

\author{Chengyan~Zhao,~\IEEEmembership{Member,~IEEE},
        Kazunori~Sakurama,~\IEEEmembership{Member,~IEEE},
        and Masaki~Ogura,~\IEEEmembership{Member,~IEEE}
\thanks{C. Zhao is with the Graduate School of Science and Engineering, Ritsumeikan University, Kusatsu, Shiga, Japan. E-mail: c-zhao@fc.ritsumei,ac.jp}
\thanks{K. Sakurama is with the Graduate School of Informatics, Kyoto University, Kyoto, Japan. E-mail: sakurama@i.kyoto-u.ac.jp}%
\thanks{M. Ogura is with the Graduate School of Information Science and Technology, Osaka University, Suita, Osaka, Japan. E-mail: m-ogura@ist.osaka-u.ac.jp}}

\markboth{Journal of \LaTeX\ Class Files,~Vol.~, No.~, ~}%
{Shell \MakeLowercase{\textit{et al.}}: Bare Demo of IEEEtran.cls for IEEE journals }

\maketitle

\begin{abstract}
This letter deals with the optimization problems of stochastic switching buffer networks, where the switching law is governed by Markov process. The dynamical buffer network is introduced, and its application in modeling the car-sharing network is also presented. To address the nonconvexity for getting a solution as close-to-the-global-optimal as possible of the optimization problem, we adopt a succinct but effective nonconvex optimization method called \emph{ DC (difference of convex functions) programming}. By resorting to the log--log convexity of a class of nonlinear functions called posynomials, the optimization problems can be reduced to DC programming problems. Finally, we verify the effectiveness of our results by simulation experiments.
\end{abstract}


\IEEEpeerreviewmaketitle


\section{Introduction}

Buffer network can well describe the dynamic process of each node and edge within a directed network, where the nodes among them behave as buffers to exchange the inflow/outflow with their neighboring nodes. The desire to conduct this research was prompted by a wide range of applications, such as water network~\cite{Michez2017}, microgrid network~\cite{Vafamand2019}, and  data transmission~\cite{Pu2018}. There is also an up-and-coming applications in the field of mobility systems, such as the 
optimal design of traffic network~\cite{Como2013,Wu2017,Grandinetti2019} and the high-efficient and effective mobility-on-demand systems~\cite{Calafiore2019,Illgen2019}. 

One of the most popular models used to analyze buffer networks is positive linear system~\cite{Farina2000}.  However, in real-life situation, neither the topology nor the parameters of the network can always remain in the time-invariant state due to the abrupt changes, such as temporal network, switched control, and time-varying parameters. To describe this phenomenon, \emph{positive Markov jump linear systems}~\cite{Bolzern2014} are proposed to efficiently develop a manipulable model for analyzing and designing, where the  law for governing the switching of each subsystem is defined by \emph{Markov process}. For stability analysis and state-feedback control of positive Markov linear systems, the relevant results can be found in~\cite{Cavalcanti2020,Zhu2017,Lian2021,Song2020} and~\cite{Li2016,Ogura2014}, respectively.

Since the past decades, the development of theoretical research on the control of dynamical networks is mainly based on positive linear systems theory~\cite{Farina2000}. An early study for investigating the stabilization problem can be found in~\cite{Ogura2014}, where the authors proposed the measurement of mean square stability and the optimization results for state-feedback control of positive Markov jump linear systems.
The authors in~\cite{Li2016} showed that the $\mathscr{L}_{\infty}$ optimization problem of positive Markov jump linear systems via state/output-feedback can be solved by \emph{linear programming}. The optimal design of Markov switching networks has been reported in~\cite{Ogura2017}, where the optimization problems are finally reduced to \emph{geometric programming problems}. 
It should be noted that, in real-life cases, state feedback is not a tractable choice. This observation has a clear reflection on the optimal design of directed network, where it is infeasible to obtain reliable measurements of state variables. In this situation, the control problem is better formulated as directly tuning the coefficients of the system matrices. Moreover, the other reason that tuning the state matrices through the state feedback only offers the freedom of adjusting row elements of state matrices. Even if the state of network control problem can be measured, a node and its outgoing edges (column vector of the the adjacency matrix) can be designed with only one variable, which is a great limitation in solving practical network optimization problems. 
To better design the buffer network, our intuition is that all elements in the adjacency matrix should be independently adjusted to achieve maximum freedom. However, we have found that if we let all elements are independently regulated, the synthesis results based on linear programming failed is transformed to be bilinear and nonconvex, which causes great difficulties in solving the problem.

To overcome this situation, we resort to a standard but efficient nonconvex optimization problem called \emph{DC (difference of convex functions) programming}~\cite{L.T.H.2005} to develop the optimization framework for positive Markov jump linear systems, which quite often gave global
solutions and proved to be more robust and more efficient than related standard methods, especially in the large-scale problem situations. DC programming has been successfully applied in various engineering areas~\cite{LeThi2018}, including mathematics finance, signals and images, and network optimization.

This letter is structured as follows. The Markov switching buffer networks with directly adjusted matrix coefficients and its application in car-sharing network are introduced in Section~\ref{network model}. 
In Section~\ref{problem}, $\mathscr{L}_1$ and $\mathscr{L}_{\infty}$ optimization problems for Markov switching buffer networks are presented. Section~\ref{main result} is devoted to the proposed results in terms of DC program. Numerical simulations are presented in Section~\ref{simulation}.


The following notations are used in this brief. Let $\mathbb{R}$, $\mathbb{R}_+$, and~$\mathbb{R}_{++}$ denote the set of real, nonnegative, and positive numbers, respectively. The set of corresponding vectors of size $n$ are denoted by $\mathbb{R}^n$, $\mathbb{R}^n_+$, and~$\mathbb{R}^n_{++}$, respectively. We let $\mathbbm{1}$ denote a column vector with all entries set to unity. The identity and zero matrices of order $n$ are denoted by $I_n$ and~$O_n$, respectively. The real matrix $A$ is said to be non-negative (positive), and is denoted by $A \geq 0$ ($A > 0$), if all entries of~$A$ are non-negative (positive). The notion $B < A$ is defined as $B-A < 0$. 
Let the Hadamard product of matrices $A$ and~$B$ be denoted by $A \odot B$. 
  We define the entry-wise exponential operation of a real vector~$v$ as $\exp[v]=[\exp v_1, \dotsc, \exp v_n]^\top$ and the entry-wise logarithm operation as $\log[v] = [\log v_1, \dotsc, \log v_n]^\top$. For a vector~$v$ with scalar entries $v_1, \dotsc, v_n$, we use $\diag(v_1, \dotsc, v_n)$ to denote the diagonal matrix. For a vector~$v$ with scalar entries $v_1, \dotsc, v_n$, we use $\diag(v_1, \dotsc, v_n)$ to denote the diagonal matrix. Let $x \in \mathbb{R}$, $\|x \|_1=\sum_{i=1}^n|x_i|$ stands for the 1-norm of a vector, whereas the vector $\infty$-norm is defined by $\|x \|_{\infty}=\max_{i \in {1,\dotsc,n}}|x_i|$. Given $v:[0, \infty) \rightarrow \mathbb{R}^n$, the $\mathscr{L}_1$-norm of a function $v(t)$ is denoted by $\|v \|_{\mathscr{L}_1}=\int_0^{\infty}\|v(t) \|_1$, and the $\mathscr{L}_{\infty}$-norm is defined by $\|v \|_{\mathscr{L}_{\infty}}=\textup{ess} \sup_{t \geq 0}\|v(t) \|_{\infty}$. $E[\cdot]$ is the mathematical expectation operator.

\section{Markov Switching buffer networks}\label{network model}
In this section, we first give the description of Markov switching buffer networks in Subsection~\ref{subsec:model}. We then present a real-life example in Subsection~\ref{subsec:car}.

\subsection{Model description}\label{subsec:model}

Consider a weighted, directed buffer network (for example, e.g.,~\cite{Rantzer2018}) defined by the graph 
$\i G=(\i V, \i E, \i W),$ where $\i V= \{{v_1, \dotsc, v_n}\}$
denotes the set of~$n$ nodes within the network and 
$
\i E = \{e_1, \dotsc, e_m\} \subseteq \i V \times \i V$
is the set of directed edges. Because the graph $\i G$ is weighted, a positive and fixed weight $w_{e_\ell}$ is assigned to an edge $e_\ell$. The scalar $w_{ij}$ is defined as the weight of the edge $(i, j)$. Thus, the adjacency matrix $A_{\i G} \in \mathbb{R}^{n \times n}$ of the graph~$\mathcal G$ is given by, 
\begin{equation*}
[A_{\i G}]_{ij}=
\begin{cases}
w_{ji},&\mbox{if $(j, i)\in \i E$},\\
0,& \mbox{otherwise}. 
\end{cases}
\end{equation*}
The set of in-neighborhoods of node $i$ is defined by  $\i N_{{i}}^{\topin}=\{{j} \in \i V: {(j, i)}\in \i E\}$. Similarly, the set of out-neighborhoods is defined by $\i N_{{i}}^{\topout}=\{{j} \in \i V: {(i, j)}\in \i E\}$. 

In this paper, we place the following assumption on the structure of the network. Suppose that there exist two special sets of nodes that serve as origins (i.e., the nodes having an empty in-neighborhood node) and destinations (i.e., the nodes having an empty out-neighborhood node). We let $\i V_{o}=\{ 1, \dotsc,  |\i V_{o}| \}$  and~$\i V_{ d}$ denote the set of origins and destinations of the buffer network, respectively. We then consider the dynamic process of the buffer network expressed by the following differential equations,
\begin{equation}\label{eq:diff}
\frac{dx_i}{dt}=\left \{
\begin{aligned}
&f_{i}^{\topin}-\sum_{j \in \i N_{{i}}^{\topout}}u_{ij}, \quad\quad\quad \  \ \text{if} \quad  i\in \i V_{o},\\
&\sum_{j \in \i N_{{i}}^{\topin}}u_{ji}-\sum_{j \in \i N_{{i}}^{\topout}}u_{ij}, \quad \text{if} \quad  i\notin \i V_{o} \cup \i V_{ d},\\
&\sum_{j \in \i N_{{i}}^{\topin}}u_{ji}-f_{i}^{\topout}, \quad\quad\quad \  \text{if} \quad  i\in \i V_{ d},
\end{aligned}
\right.
\end{equation}
where {$x_i$~($i=\{1, \dotsc, n\}$)} represents buffer variable in node~$i$, $u_{ij}$ is the volume of flow from node $i$ to $j$, and~$f_{i}^{\text{in}}$ and~$f_{i}^{\text{out}}$ denote the inlet and outlet effects, respectively. 

In this brief, the flows among the buffer network are assumed to obey the following linear form:
\begin{equation}\label{eq:flow}
f_{i}^{\text{out}}=\beta_ix_i, \ \ u_{ij}=\delta_{ij}w_{ij}x_i,
\end{equation}
where $\beta=\{\beta_i\}_{i \in \i V_{ d}}$  and  $\delta=\{\delta_{ij}\}_{(i, j)\in \i E}$ are the parameters to be tuned in the next section. Herein, if we adopt the results of state/output-feedback in~\cite{Li2016}, the freedom of tuning the flow $u_{ij}=\phi_{i}w_{ij}x_i$ only derives from the parameter $\phi_{i}$. In this letter, we allow $u_{ij}$ to be designed independently on the parameter of each edge as shown in~\eqref{eq:flow}.

To measure the performance of the buffer network, we adopt the 
output 
$y = [x^\top \ \alpha u^\top]^\top$, 
where $\alpha > 0$ is a weight constant and~$u \in \mathbb{R}_+^{n \times n}$ includes the information of the edges. 
If we define the matrix $B$ and $D$ by 
\begin{equation*}
B_{ij}=
	\begin{cases}
	\beta_{i},&\mbox{if $i = j$}, \\
	0,&\mbox{otherwise},
	\end{cases}\quad
	D_{ij}=
	\begin{cases}
	\delta_{ji},&\mbox{if $(i, j)\in \i E$}, \\
	0,&\mbox{otherwise},
	\end{cases}
\end{equation*}
then the dynamic model can then be expressed as 
\begin{equation*}
\Sigma: \left \{
\begin{aligned}
&\dot{x}=\Bigl(D \odot  A_{\i G} -\diag\bigl(\mathbbm{1}^\top (D \odot A_{\i G})\bigr)-B \Bigr)x+G^{\topin} f^{\topin},\\
&y=G^{\topout}(\delta)x,
\end{aligned}
\right.
\end{equation*}
where input vector $f^{\topin}$,  input matrix $G^{\topin}$, and output matrix $G^{\topout}(\delta)$ are defined by  $f^{\topin}=[f_1^{\topin} \cdots f_{|\i V_{o}|}^{\topin}]^\top$ and 
\begin{equation} \label{eq:B}
 G^{\topin}=
\begin{bmatrix}
I_{|\i V_{o}|}\\
O_{n-|\i V_{o}|, |\i V_{o}|}
\end{bmatrix},~
G^{\topout}(\delta)= 
\begin{bmatrix}
I_n\\
\alpha H(\delta)
\end{bmatrix}.
\end{equation}
The matrix $H(\delta)$ is defined by 
$    H(\delta)_{\ell i} =   w_{e_\ell}$ if $i = e_\ell(1)$ and $ H(\delta)_{\ell i} = 0$ otherwise.
For each edge $e_\ell$, we use the notation $e_\ell = (e_\ell(1), e_\ell(2))$, wherein the nodes $e_\ell(1)$ and~$e_\ell(2)$ denote the origin and destination of the edge, respectively. Since $G^{\topin}$ and $G^{\topout}(\delta)$ are nonnegative matrices and $D \odot  A_{\i G} -\diag\bigl(\mathbbm{1}^\top (D \odot A_{\i G})\bigr)-B$ is the Metzler matrix, following the definition in~\cite{Farina2000}, the dynamic model~\eqref{eq:B} is the positive linear system. Furthermore, if the adjacent matrix $A_{\i G}$, input matrix $G^{\topin}$, and output matrix $G^{\topout}(\delta)$ in~\eqref{eq:B} varies under Markov process and $A=D \odot  A_{\i G} -\diag\bigl(\mathbbm{1}^\top (D \odot A_{\i G})\bigr)-B$, system $\Sigma$ is upgraded to Markov switching buffer networks:
\begin{equation}\label{eq:sys}
\Sigma_\sigma: \left\{ 
\begin{aligned}
&\dot{x}=A_{\sigma(t)}(\beta ,\delta)x+G^{\topin}_{\sigma(t)} f^{\topin},  \\
&y=G^{\topout}_{\sigma(t)}(\delta)x,
\end{aligned}
\right.
\end{equation}
 where $\sigma=\{\sigma(t)\}_{t \geq 0}$ is a time-homogeneous Markov process taking values in the finite discrete set $S=\{1, \dotsc, N\}$. 
 We assume that $\Sigma_\sigma$ is positive if the subsystems ($A_i(\beta ,\delta)$, $G^{\topin}_i$, and $G^{\topout}_i(\delta)$) are positive for all $i \in S$. The probability rate matrix of system $\Sigma_\sigma$ is given by
 \begin{equation*}
\textup{Pr}\{\sigma(t+h)=j \mid \sigma(t)=i\}=
\left\{ 
\begin{aligned}
&\pi_{ij}h+o(h),  \quad~~~  \mbox{if } j\neq i,\\
&1+\pi_{ii}h+o(h),~ \mbox{if } j = i,
\end{aligned}
\right.
\end{equation*}
where $\pi_{ij} > 0$, the parameterized transition rate from mode $i$ to mode $j$ obeys the equations 
\begin{equation}\label{prob:}
   \pi_{ii}+\sum_{j=1,i\neq j}^N\pi_{ij}=0 
\end{equation}
and $o(h)$ is little-$o$ notation defined by $\lim_{h \to 0}o(h)/h=0$.


\subsection{Example: car-sharing service network}\label{subsec:car}

In the one-way car-sharing service network, wherein the stations providing parking slots for customers renting/returning vehicles at any stations. Despite the convenience, this service has the shortcoming of uneven distribution of vehicles as the service proceeds, which causes parking slots or vehicles to be unavailable at particular stations. To reduce the uneven distribution, dynamics pricing is promising, which controls
the demand of customers by adjusting usage prices in real-time. Here, we discuss how we can determine the prices for efficient control.

First, we construct a mathematical model of the system of the one-way car-sharing service according to the authors' previous paper~\cite{Ikeda2021}. Let $n$ be the number of stations of the service network which corresponds to a node in a weighted and directed graph $\mathcal{N}$. Let $x_i$ ($i \in \mathcal{N}$) be the expectation of the number of vehicles
parking at station $i$.
The possible usage between stations is described by an edge set $\mathcal{E}$.
Let $u_{ij}$ ($(i,j) \in \mathcal{E}$) be 
the expectation of the number of customers who travel from station $i$ to $j$
within a time interval.
Let $f_i^{\mathrm{in}}$ 
($f_i^{\mathrm{out}}$)
be the expectation of the number of vehicles moving 
to (from) this area from (to) other areas.
Then, this system is modeled as equation~\eqref{eq:diff}.

Next, we construct a model of the demand of customers which can change with prices.
Assume that the expectation of the demand is $\bar u_{ij}(t)$
when the price is $\bar p_{ij}(t)$
and that the change of the demand is governed with an affine model
around this point.
Let $p_{ij}$ be the price
for traveling from stations $i$ to $j$,
and let $\delta_{ij}$ be the price elasticity.
Then, the affine model is given as
\begin{equation}
 u_{ij} = \bar u_{ij} - \delta_{ij} (p_{ij} - \bar p_{ij})
 \label{demand_model}
.
\end{equation}
As a pricing strategy, the price $p_{ij}$ is adjusted
according to the number $x_i$ of vehicles at station $i$
as follows:
\begin{equation}
 p_{ij} = \hat p_{ij} - w_{ij} x_i
 \label{price_model}
,
\end{equation}
where $\hat p_{ij}$ and $w_{ij}$ are design parameters.
We set $\hat p_{ij} = \bar p_{ij} + \bar u_{ij}/\delta_{ij}$,
and from (\ref{demand_model}) and (\ref{price_model}),
the demand model is reduced to
 $u_{ij} = \bar u_{ij} - \delta_{ij} (p_{ij} - \bar p_{ij})
           = \bar u_{ij} - \delta_{ij} (\hat p_{ij} - w_{ij} x_i - \bar p_{ij})
           = \delta_{ij} w_{ij} x_i$. 
This corresponds to the equation~\eqref{eq:flow}.

In practical service network~\cite{Bai2014}, to effectively deal with uncertainties for obtaining low expected operational costs, stochastic service network model is proposed to suit all possible future scenarios arise in practice. Thus, model~\eqref{eq:sys} is concise and feasible in describing the stochastic switching networks. With these preparations, it is feasible to design parameters $w_{ij}$ by applying the proposed method in the following section.




\section{Problem formulation}\label{problem}

Following the formulation in the previous section, we assume that the decay rates of node $i$ and directed edge $ij$ can be tuned by the parameters $\beta_i$ and~$\delta_{ij}$ to improve the performance of the buffer network. Calculation of the sum of all variables yields the cost function
\begin{equation}
    L(\beta, \delta)=\sum_{i\in \i V_{\mathnormal{d}}}g_i(\beta_i)+\sum_{(i, j)\in \mathcal E} h_{ij}(\delta_{ij}),
\end{equation}
where the variables are tuned within the following intervals 
\begin{equation}\label{eq:interval}
0<\beta_i \leq \bar{\beta}_i,\ 0<\delta_{ij} \leq \bar{\delta}_{ij}.
\end{equation}
In this letter, we adopt $\mathscr{L}_1$ and $\mathscr{L}_{\infty}$ norms as the requirements of the buffer network.
For an exponential mean stable system $\Sigma_\sigma$, and the initial condition $\sigma(0)$ and $w \in \mathscr{L}_1$, if there exsits $\gamma >0$ such that $\|E[y]\|_{\mathscr{L}_1} \leq \gamma\|u\|_{\mathscr{L}_1}$. The $\mathscr{L}_1$-gain of system $\Sigma_\sigma$, denoted by $\|\Sigma_\sigma\|_1$, is defined by
\begin{equation*}
    \|\Sigma_\sigma\|_1=\sup_{u \in \mathscr{L}_1}\frac{\|E[y]\|_{\mathscr{L}_1}}{\|u\|_{\mathscr{L}_1}}.
\end{equation*}
 Likewise, if there exists $\gamma >0$ such that $\|E[y]\|_{\mathscr{L}_{\infty}} \leq \gamma\|u\|_{\mathscr{L}_{\infty}}$. The $\mathscr{L}_{\infty}$-gain of system $\Sigma_\sigma$, denoted by $\|\Sigma_\sigma\|_{\infty}$, is defined by
\begin{equation*}
     \|\Sigma_\sigma\|_{\infty}= \sup_{u \in \mathscr{L}_{\infty}}\frac{\|E[y]\|_{\mathscr{L}_{\infty}}}{\|u\|_{\mathscr{L}_{\infty}}}.
\end{equation*}


We are now ready to state the optimization problems studied in this letter.

\begin{problem}\label{eq:prob_H2}($\mathscr{L}_1$/$\mathscr{L}_{\infty}$-norm optimization): 
Given the desired parameter tuning cost $\bar{L}(\beta, \delta)$, find the parameters $\beta$ and $\delta$ minimizing the $\mathscr{L}_1$/$\mathscr{L}_{\infty}$ norm, under the constraint that the parameter constraints~\eqref{eq:interval} are satisfied.
	\\
\end{problem}
\begin{problem}\label{eq:prob_H_infty}($\mathscr{L}_1$/$\mathscr{L}_{\infty}$-norm constrained optimization): Given the desired $\mathscr{L}_1$/$\mathscr{L}_{\infty}$ norm, find the parameters $\beta$ and $\delta$ minimizing the parameter tuning cost $\bar{L}(\beta, \delta)$, under the constraint that the parameter constraints~\eqref{eq:interval} are satisfied.
	\\
%
\end{problem}

The difficulty of solving Problem~\ref{eq:prob_H2} and Problem~\ref{eq:prob_H_infty} mainly stems from the noncovnexity through independently tuning the edges weight $\delta_{ij}$ of adjacency matrices.


\section{Main results}\label{main result}

In this section, we present the solutions to Problem~\ref{eq:prob_H2} and Problem~\ref{eq:prob_H_infty} in terms of DC programming~\cite{Horst1999,LeThi2018}.
We begin this section by introducing the preliminary knowledge of posynomials, DC functions, and DC program to derive the main results. 

\begin{definition}\label{def:posy}
Let $v_1$, $\dotsc$, and~$v_n$ denote $n$ real positive variables. We state that a real function $g(v)$  is a {\it monomial} if $c>0$ and~$a_1, \dotsc, a_n \in \mathbb{R}$ such that $g(v) = c v_{\mathstrut 1}^{a_{1}} \dotsm v_{\mathstrut n}^{a_n}$. We state that a real function $f(v)$ is a {\it posynomial}~\cite{Boyd2007} if $f$ is the sum of the monomials of~$v$.
\end{definition}

The following lemma shows the log-convexity of posynomials~\cite{Boyd2007}.
\begin{lemma}\label{log:convexity}
	If $f \colon \mathbb{R}_{++}^{\mathnormal n} \to \mathbb{R}_{++}$ is a posynomial function,  then the log-transformed function
	$$F \colon \mathbb{R}^{\mathnormal n} \to \mathbb{R} \colon \mathnormal w \mapsto \log[f(\exp[w])]$$
	is convex. 
\end{lemma}

The property shown in Lemma~\ref{log:convexity} enables us to build a relationship with a general class of mathematical programming that deals with the difference between two convex functions, called \emph{DC programming}. 

\begin{definition}(DC functions~\cite{Horst1999})
 Let $\mathcal{C}$ be a convex subset of~$\mathbb{R}^{\mathnormal{n}}$. A real-valued functions $f:\mathcal{C} \rightarrow \mathbb{R}$ is called a {\it DC function} on $\mathcal{C}$ if there exist two convex functions $g,h: \mathcal{C} \rightarrow \mathbb{R}$ such that $f$ can be expressed in the form
 \begin{equation*}
     f(x)=g(x)-h(x).
 \end{equation*}
\end{definition}

In principle, every continuous function can be approximated by a DC function with the desired precision. Based on decomposition methods~\cite{Horst1999}, it is possible to convert a nonconvex optimization problem to a DC programming problem.
\begin{definition}(DC programming problem~\cite{Horst1999})\label{DC programming}
Programming problems dealing with DC functions are called {\it DC programming problems}. Let $\mathcal{C}$ be a closed convex subset of~$\mathbb{R}^{ n}$, and the general form of DC programming problem considered in this brief is 
\begin{align*}
	\minimize_{x\in \mathcal{C}}\ &f_0(x)
	\\
\subjectto\ & f_i(x)\leq 0,\,i=1,\dots,m,
\end{align*}
	where $f_0(x)=g_0(x)-h_0(x)$ and~$f_i(x)=g_i(x)-h_i(x)$ are the differences of the two convex functions.
\end{definition}
 To optimally solve the DC programming problem, we can choose the branch-and-bound type and outer-approximation algorithms~\cite{Oliveira2018}, which lead to more efficient procedures. Subject to the proper assumption for the cost function, we show that Problems~\ref{eq:prob_H2} and~\ref{eq:prob_H_infty} can be transformed to DC programming problems.

\begin{asm}\label{costfunc}
The following assumptions are made:
\begin{enumerate}
\item The set of system matrices $A_i(\beta, \delta)$
\begin{equation*}
  A_i(\beta, \delta)=A^{o}_i(\delta)+A^{d}_i(\beta, \delta), i \in \{1, \dotsc, N\},
\end{equation*}
where $A^{o}_i(\delta)$ and $A^{d}_i(\beta, \delta)$ are matrices with posynomial or zero entries.

\item The functions $g_i(\beta_i)$ and~$h_{ij}(\delta_{ij})$ are posynomials for all $i$ and $j$. 

\end{enumerate}
\end{asm}


\begin{theorem}\label{theorem_H2}
Under Assumption~\ref{costfunc}, if we make $\mathscr{L}_1$ norm of system $\Sigma_\sigma$ to be minimized, the solution of Problem~\ref{eq:prob_H2} is given by the solution of the following DC programming problem
\begin{small}
\begin{align*}
\minimize_{\mathclap{\gamma > 0, v_i \in \mathbb {R}^{n}_+,
\atop
\phi, \eta \in \mathbb{R}}} ~ &\gamma
\\ 
\subjectto\ \ 
&\log[v_i^\top A^o_i(\exp[\eta]) +\Sigma_{i\neq j}^N\pi_{ij}v_j^\top+\mathbbm{1}_r^\top G_i^{\topout}(\exp[\eta])^\top]  \notag  \\  
& - \log [v_i^\top A^{d}_i(\exp[\eta],\exp[\phi]) -\pi_{ii}v_i^\top] \leq 0,   \\
&	\log[v_i^\top G_i^{\topin}]  -\log[\gamma\mathbbm{1}_s^\top] \leq 0,  \\
	    &\log[L(\exp{[\eta]}, \exp[\phi])] - \log[\bar L] \leq 0,\\
&\log[ \exp[\phi]]- \log[\exp[\bar{\phi}]]\leq 0,\\ 
&\log[\exp[\eta]]- \log [\exp[\bar{\eta}]]\leq 0.
\end{align*}
The solution of Problem~\ref{eq:prob_H2} is then given by
\begin{equation}\label{solution}
B=\exp[\phi],\ D=\exp[\eta].
\end{equation}
\end{small}
\end{theorem}


 \begin{proof}
 Based on Section~\ref{network model}, system~\eqref{eq:sys} is proved to be a standard positive Markov jump linear systems. Resorting to the stability result in Theorem 4~\cite{Ogura2017}, we can show that if the buffer network~$\Sigma_\sigma$ is exponential mean stable, then the following problem is equivalent to Problem~\ref{eq:prob_H2}: 
 \begin{small}
 \begin{subequations}\label{eq:proof_H2}
 \begin{align}
 \minimize_{\mathclap{\beta, \delta \in\Theta, v \in \mathbb{R}_{+}^{n} }}\ \ &\gamma \label{eq:proof_H2:A}\\
 \subjectto\ \ & v_i^\top A_i(\beta, \delta) +\Sigma_{i=1}^N\pi_{ij}v_j^\top+\mathbbm{1}_r^\top G_i^{\topout}(\delta) \leq 0, \label{eq:proof_H2:B}
 \\
 & 	v_i^\top G_i^{\topin}  -\gamma\mathbbm{1}_s^\top \leq 0, \label{eq:proof_H2:C}\\
 & L(\beta, \delta) \leq \bar{L}, \label{eq:proof_H2:D}\\
 & \eqref{eq:interval}.\label{eq:proof_H2:E}
 \end{align}
 \end{subequations}
 \end{small}
 From the observation in~\eqref{eq:B}, $A_{i}(\beta, \delta)$ includes of posynomials within negative sign, therefore, \eqref{eq:proof_H2} is no longer a linear programming and become nonconvex.
 For this situation, we resort to the log--log convexity of the posynomials in Lemma~\ref{log:convexity} for reducing~\eqref{eq:proof_H2} into DC programming problem. According to Definition~\ref{def:posy} and Assumption~\ref{costfunc}, the sum of~$g_i$ and~$h_{ij}$ in~\eqref{eq:proof_H2:D} is the sum of monomials in essence. Thus, \eqref{eq:proof_H2:D} is also a posynomial function. The object function~\eqref{eq:proof_H2:A} subtracts $0$ that satisfies the DC functions in Definition~\ref{DC programming}. For the constraint~\eqref{eq:proof_H2:C}, the product of constant matrices $G_i^{\topin}$ and positive vector variables are nonnegative matrices with posynomial entries. Since $\gamma\mathbbm{1}_s^\top$ is obviously nonnegative, \eqref{eq:proof_H2:C} is the difference between two nonnegative matrices, which can be successfully transformed to DC functions by the log--log transformation. According to Assumption~\ref{costfunc}, $A_{i}(\beta, \delta)$ is decomposed into two nonnegative matrices, where $A^{o}_i(\delta)= D \odot  A_{\i G}$ and  $A^{d}_i(\beta, \delta)=\mathbbm{1}^\top (D \odot A_{\i G})+B$. From the model description in Section~\ref{subsec:model}, each entries among the decomposed matrices are either posynomials or zero.
Likewise, the same way for the decomposition of probability rate matrix.
 Thus, \eqref{eq:proof_H2:B} is equivalent to 
 $(v_i^\top A^o_{i}(\delta)+\Sigma_{i\neq j}^N\pi_{ij}v_j^\top+\mathbbm{1}_r^\top G_i^{\topout}(\delta))-(v_i^\top A^{d}_i(\beta, \delta) -\pi_{ii}v_i^\top )< 0$,
 which shows the difference between the two materials. Similarly, \eqref{eq:proof_H2:B} can also be transformed to DC functions. For~\eqref{eq:proof_H2:E}, we can directly obtain the variable constraints in the form of DC functions from~\eqref{eq:interval}. Hence, Theorem~\ref{theorem_H2} is a DC programming problem. This completes the proof of theorem. 
 \end{proof}

%
 
Theorem~\ref{theorem_H2} shows that Problem~\ref{eq:prob_H2} can be
turned into an equivalent DC program. In solving Problem~\ref{eq:prob_H_infty}, we adopt $\mathscr{L}_{\infty}$ norm as the performance constraint of system $\Sigma_\sigma$ to show the framework of minimizing the parameter tuning cost. 

\begin{theorem}\label{thm:H_infty}
Under the aforementioned assumptions and lemma, if we set a $\mathscr{L}_{\infty}$-norm constraint for system $\Sigma_\sigma$ by $\bar{\gamma}>0$, the solution of Problem~\ref{eq:prob_H_infty} is given by the solution of the following DC programming problem,
\begin{small}
\begin{align*}
\minimize_{\mathop{\phi, \eta \in \mathbb{R}
\atop
v_i \in \mathbb {R}^{n}_+   
}}
\ \ & \log[ L(\exp[\phi], \exp[\eta])]
\\
\subjectto\ \ \ &\log[ A^o_i(\exp[\eta])v_i +\Sigma_{i\neq j}^N\pi_{ij}v_j^\top+G_i^{\topin}\mathbbm{1}_r ]  \notag  \\  
& - \log[ A^{d}_i(\exp[\eta],\exp[\phi])v_i -\pi_{ii}v_i^\top] \leq 0,   \\
&	\log[G_i^{\topout}v_i]  -\log[\bar{\gamma}\mathbbm{1}_s^\top] \leq 0,  \\
&\log[ \exp[\phi]]- \log[\exp[\bar{\phi}]]\leq 0,\\ 
&\log[\exp[\eta]]- \log [\exp[\bar{\eta}]]\leq 0.
\end{align*}
\end{small}
The solution of Problem~\ref{eq:prob_H_infty} is then given by~\eqref{solution}.
\end{theorem} 
\begin{proof}
Relying on the $\mathscr{L}_{\infty}$ stability results in Theorem 2~\cite{Li2016} which shows that if the positive linear system~$\Sigma_\sigma$ is internally stable and~$\|\Sigma_\sigma\|_\infty < \bar{\gamma}$, we can show that the following optimization problem is equivalent to Problem~\ref{eq:prob_H_infty}: 

\begin{subequations}\label{eq:proof_H_infty}
\begin{small}
 \begin{align}
 \minimize_{\mathop{\beta, \delta \in\Theta, v \in \mathbb{R}_{+}^{n}}}\ \ & L(\beta, \delta) \\
 \subjectto \ \ &  A_i(\beta, \delta)v_i +\Sigma_{i=1}^N\pi_{ij}v_j^\top+G_i^{\topin}\mathbbm{1}_s  \leq 0,
 \\
 & 	G_i^{\topout}(\delta) v_i  -\bar{\gamma}\mathbbm{1}_r \leq 0,\\
 & \eqref{eq:interval}.
 \end{align}
 \end{small}
 \end{subequations}
 
 The direction of the proof of the equivalence between Problem~\ref{eq:prob_H_infty} and~\eqref{eq:proof_H_infty} is the same as the procedure of Theorem~\ref{theorem_H2} and, therefore, is omitted.
 \end{proof}

\section{Numerical simulation}\label{simulation}

\section{Conclusion}

\section*{Acknowledgement}
This work was partially supported by the joint project of Kyoto University and Toyota Motor Corporation, titled ``Advanced Mathematical Science for Mobility Society.''


\begin{thebibliography}{99}
 \bibitem{Michez2017}
A. Michez et al.,
“Multi-temporal monitoring of a regional riparian buffer network (12,000 km) with LiDAR and photogrammetric point clouds,” {\it Journal of Environmental Management},
vol. 202, pp. 424--436, 2017.

\bibitem{Vafamand2019}
N. Vafamand et al., “Networked fuzzy predictive control of power buffers for dynamic stabilization of DC microgrids,” {\it IEEE Transactions on Industrial Electronics}, vol. 66, no. 2, pp. 1356--1362, 2019.

\bibitem{Pu2018}
C. Pu et al., “Bufferless transmission in complex networks,” {\it IEEE Trans. Circuits Syst. II: Express Briefs}, vol. 65, no. 7, pp. 893--897, 2018.

\bibitem{Como2013}	
 G. Como et al., “Stability analysis of transportation networks with multiscale driver decisions,” {\it SIAM Journal on Control and Optimization}, vol. 51, pp. 230--252, 2013.

\bibitem{Wu2017}
J. Wu and Y. Xia, “Complex-network-inspired design of traffic generation patterns in communication networks,” {\it IEEE Transactions on Circuits and Systems II: Express Briefs}, vol. 64, no. 5, pp. 590--594, 2017.

\bibitem{Grandinetti2019}
P. Grandinetti, C. Canudas-de-Wit, and F. Garin, “Distributed optimal traffic lights design for large-scale urban networks,” {\it IEEE Transactions on Control Systems Technology}, vol. 27, no. 3, pp. 950--963, 2019.

\bibitem{Calafiore2019}
G. C. Calafiore, C. Bongiorno, and A. Rizzo, “A robust MPC
approach for the rebalancing of mobility on demand systems,”  {\it Control Engineering Practice}, vol. 90, pp. 169--181, 2019. 



\bibitem{Illgen2019}
S. Illgen and M. Hock, “Literature review of the vehicle relocation problem in one-way car sharing networks,”  {\it Transportation Research Part B: Methodological.}, vol. 120, pp. 193--204, 2019. 


	
 
 	 \bibitem{Farina2000}
 L. Farina and S. Rinaldi, “{\it Positive Linear Systems: Theory and Applications,}” John Wiley, 2000.

\bibitem{Bolzern2014}
P. Bolzern, P. Colaneri, and G. D. Nicolao, Stochastic stability of Positive Markov Jump Linear Systems, Automatica, Volume 50, Issue 4, Pages 1181-1187, 2014.

\bibitem{Cavalcanti2020}
Joao Cavalcanti, Hamsa Balakrishnan,
Sign-stability of Positive Markov Jump Linear Systems,
Automatica,
Volume 111,
2020,


\bibitem{Zhu2017}
S. Zhu, Q. Han and C. Zhang, "$L_1$-Stochastic Stability and $L_1$-Gain Performance of Positive Markov Jump Linear Systems With Time-Delays: Necessary and Sufficient Conditions," in IEEE Transactions on Automatic Control, vol. 62, no. 7, pp. 3634-3639, July 2017, doi: 10.1109/TAC.2017.2671035.


\bibitem{Lian2021}
Jie Lian, Renke Wang,
Stochastic stability of positive Markov jump linear systems with fixed dwell time,
Nonlinear Analysis: Hybrid Systems,
Volume 40,
2021,
101014,
\bibitem{Song2020}
Song X, Lam J, Chen X, et al. Descriptor state‐bounding observer design for positive Markov jump linear systems with sensor faults: Simultaneous state and faults estimation[J]. International Journal of Robust and Nonlinear Control, 2020, 30(5): 2113-2129.

\bibitem{Li2016}
Shuo Li, Zhengrong Xiang,
Stochastic stability analysis and $L_{\infty}$-gain controller design for positive Markov jump systems with time-varying delays,
Nonlinear Analysis: Hybrid Systems,
Volume 22,
2016,
Pages 31-42,
\bibitem{Ogura2014}
	M. Ogura and C. F. Martin, “Stability analysis of positive semi-Markovian jump linear systems with state resets,” SIAM Journal on Control and Optimization, vol. 52, pp. 1809-1831, 2014. 
	

 

\bibitem{Ogura2017} M. Ogura and V. M. Preciado, Optimal design of switched networks of positive linear systems via geometric programming, {\it IEEE Transactions on Control of Network Systems}, vol. 4, no. 2, pp. 213-222, 2017.
 
 \bibitem{L.T.H.2005}
An, L.T.H., Tao, P.D. The DC (Difference of Convex Functions) Programming and DCA Revisited with DC Models of Real World Nonconvex Optimization Problems. Ann Oper Res 133, 23–46 (2005).

 \bibitem{LeThi2018}
 H.A. Le Thi, T. Pham Dinh, “DC programming and
DCA: thirty years of developments,” {\it Mathematical Programming}, vol. 169, pp. 5--68, 2018.

 \bibitem{Rantzer2018}
 A. Rantzer and M. E. Valcher, “A tutorial on positive systems and large scale control,”
 in {\it 57th IEEE Conf. Dec. and Con.}, pp. 3686--3697, 2018.
 
 
 










\bibitem{Ikeda2021}
T. Ikeda, K. Sakurama, and K. Kashima, “Multiple sparsity constrained control node scheduling with application to rebalancing of mobility networks,” {\it IEEE Transactions on Automatic Control}, doi: 10.1109/TAC.2021.3115441.

\bibitem{Bai2014}
Ruibin Bai, Stein W. Wallace, Jingpeng Li, Alain Yee-Loong Chong,
Stochastic service network design with rerouting,
Transportation Research Part B: Methodological,
Volume 60,
2014,











\bibitem{Horst1999}	
 R. Horst and N. V. Thoai, “DC programming: overview,” {\it Journal of Optimization Theory and Applications}, vol. 103, pp. 1--43, 1999.
 
 \bibitem{Boyd2007}
 S. Boyd, S.J. Kim, L. Vandenberghe, and A. Hassibi, “A tutorial on geometric programming,”
 {\it Opt. and Eng.}, vol. 8, pp. 67--127, 2007.
 



with bounded controls,” {\it IEEE Transactions on Circuits and Systems II: Express Briefs}, vol. 54,
no. 2, pp. 151--155, 2007.










\bibitem{Oliveira2018}	
 W. Oliveira, “Proximal bundle methods for nonsmooth DC programming,” {\it Journal of Global Optimization}, vol. 75, pp. 523--563, 2019.


\end{thebibliography}


\bibliographystyle{unsrt}   

\end{document}